%% file: main.tex
\documentclass[10pt,twocolumn]{article}

\usepackage[letterpaper,textwidth=6.75in,textheight=9.25in,centering]{geometry}
\usepackage{amsmath,amssymb,amsfonts,amsthm,mathtools}
\usepackage{booktabs}
\usepackage{graphicx}
\usepackage{fontawesome5}
\usepackage{microtype}
\usepackage[font=small,labelfont=bf]{caption}
\usepackage[round]{natbib}

\usepackage{tikz}
\usetikzlibrary{arrows.meta,calc,decorations.pathreplacing,positioning,shadows}
\usepackage[colorlinks=true,citecolor=blue,linkcolor=blue,urlcolor=blue,
  hyperfootnotes=false]{hyperref}

\newtheorem{theorem}{Theorem}[section]
\newtheorem{lemma}[theorem]{Lemma}
\newtheorem{proposition}[theorem]{Proposition}
\newtheorem{corollary}[theorem]{Corollary}
\theoremstyle{definition}
\newtheorem{definition}[theorem]{Definition}
\newtheorem{remark}[theorem]{Remark}

\newcommand{\E}{\mathbb E}
\newcommand{\Pp}{\mathbb P}
\newcommand{\R}{\mathbb R}
\newcommand{\1}{\mathbf 1}
\newcommand{\eps}{\epsilon}
\newcommand{\loc}{\mathrm{loc}}

\newcommand{\dist}{\operatorname{dist}}
\newcommand{\Var}{\operatorname{Var}}
\newcommand{\Unif}{\operatorname{Unif}}
\newcommand{\sgn}{\operatorname{sgn}}
\newcommand{\mom}{\mathcal D}

\hypersetup{
  pdftitle={Universal Refinement without Interaction: Order-Optimal 1-Bit Mean Estimation},
  pdfauthor={Yuchen Miao},
  pdfsubject={Fully non-adaptive 1-bit mean estimation under finite moments}
}

\title{Universal Refinement without Interaction:\\
Order-Optimal 1-Bit Mean Estimation}
\author{Yuchen Miao\\
\small Northeastern University, China\\
\small\texttt{miaoyc@mails.neu.edu.cn}\\[0.8ex]
\small\faGithub\enspace\url{https://github.com/PNemo04/onebitmean}}
\date{}

\begin{document}
\maketitle
\vspace{-1.5em}

\begin{abstract}
This paper shows that interaction is unnecessary for order-optimal 1-bit mean
estimation under finite central moments.  For distributions satisfying
$|\E X|\leq\lambda$ and $\E|X-\E X|^k\leq\sigma^k$ for a fixed $k>1$, we
construct a fully non-adaptive public-coin protocol that fixes every measurable
1-bit query before communication.  All localization and refinement queries are
generated in a single batch; a subsequently decoded coarse center changes
only how the stored refinement bits are interpreted.  Two complementary
constructions realize this decoder-side refinement: a finite dyadic scheme
based on periodic residues and a continuous-scale scheme based on shifted
random grids.  Up to $k$-dependent constants, the refinement cost is
$(\sigma/\eps)^2\log(1/\delta)$ for $k>2$,
$(\sigma/\eps)^2[1+\log(\sigma/\eps)]\log(1/\delta)$ for $k=2$, and
$(\sigma/\eps)^{k/(k-1)}\log(1/\delta)$ for $1<k<2$.  Together with the
additive localization cost $1+\log(\lambda/\sigma)$, these rates answer the
Lau--Scarlett open problem for arbitrary measurable 1-bit queries in the
affirmative.  In the parameter range covered by existing small-error,
high-confidence lower bounds, the resulting sample complexity is minimax
optimal.
\end{abstract}

\section{Introduction}

Mean estimation is a basic primitive in statistics, learning, and distributed
optimization.  In a decentralized system, however, the learner may never see
the raw observations: each device observes one fresh sample and communicates
only a compressed message.  This paper considers the extreme 1-bit regime,
in which independent devices observe $X_1,\ldots,X_n$ from an unknown
distribution on $\R$, each device returns a single Boolean answer, and the
learner must estimate the unknown mean from these $n$ bits.

If the learner may choose each query after seeing earlier answers, adaptive
threshold procedures attain the optimal finite-moment rates
\citep{lau2026order}.  Their optimality, however, comes with systems costs
beyond the number of transmitted bits.  Sequential rounds create latency,
require devices to remain available, and prevent all queries from being
deployed in a single batch.  These limitations motivate the open problem
posed by \citet[Open Problem~1]{lau2026open} at COLT 2026:
\par\vspace{0.35\baselineskip}
\noindent\hfill
\begin{minipage}{0.94\columnwidth}
\centering\fontfamily{ppl}\selectfont\itshape
Can fully non-adaptive arbitrary 1-bit queries attain the same minimax rates
as adaptive protocols for nonparametric mean estimation under only a bounded
central moment?
\end{minipage}
\hfill\null\par\vspace{0.35\baselineskip}
The challenge is that every non-adaptive query must remain useful before the
learner knows even a coarse location of the mean.  For threshold or interval
queries, this requirement creates a genuine minimax gap between adaptive and
non-adaptive protocols \citep{lau2026order}.  We show that arbitrary
measurable 1-bit queries close this gap: a fully non-adaptive protocol
attains the same order-optimal sample complexity as adaptive protocols.

\subsection{Problem setup and main question}

Fix a known moment order $k>1$.  For known $\lambda\geq\sigma>0$, define the
nonparametric class
\[
 \mom(k,\lambda,\sigma)
 =\left\{P:
 \begin{array}{l}
 |\E_P X|\leq\lambda,\\[-1mm]
 \E_P|X-\E_PX|^k\leq\sigma^k
 \end{array}
 \right\}.
\]
Apart from these two constraints, $P$ may be discrete, asymmetric, and have
unbounded support.  Write $\ell(x)=1+\log x$ for $x\geq1$, so every
logarithmic factor is at least one.\footnote{The alternatives
$\max\{1,\log x\}$ and $\log(1+x)$ are equivalent up to constant factors.}
Here $\lambda/\sigma$ measures
coarse positional uncertainty, whereas $(k,\sigma)$ govern residual tails
after localization.  Thus $\lambda$ contributes only an additive localization
cost, while $k$ determines the accuracy exponent.  The target 1-bit minimax
rate is
\begin{equation}
\label{eq:colt-rate}
 \mathfrak r_k(\lambda,\sigma,\eps,\delta)
 :=\ell(\lambda/\sigma)+
 \begin{cases}
 \dfrac{\sigma^2}{\eps^2}\log\dfrac{1}{\delta},
 & k>2,\\[2mm]
 \dfrac{\sigma^2}{\eps^2}\ell(\sigma/\eps)
 \log\dfrac{1}{\delta},
 & k=2,\\[2mm]
 (\dfrac{\sigma}{\eps})^{\frac{k}{k-1}}
 \log\dfrac{1}{\delta},
 & 1<k<2.
 \end{cases}
\end{equation}

\begin{definition}[Fully non-adaptive public-coin protocol]
\label{def:nonadaptive}
A public seed\footnote{Equivalently, a public-coin protocol is a distribution
over deterministic non-adaptive query lists; conditioning on the seed fixes
the entire list.} is drawn independently of the data and, before any message
is observed, determines measurable query sets
$A_1,\ldots,A_n\subseteq\R$.  Device $t$ observes only its fresh sample
$X_t$ and transmits $Y_t=\1\{X_t\in A_t\}$.  The decoder may use the full
bit transcript and the public seed, but no query set may depend on any
transmitted bit.  Equivalently, conditional on the public seed and the known
parameters, every $A_t$ is fixed independently of the transcript: a quantity
decoded later, such as a coarse center $c$, may affect only the decoding map
and never the set whose indicator produced $Y_t$.
\end{definition}

The target is an estimator
$\widehat\mu$, computed from the transcript and the public randomness, such
that
\[
 \sup_{P\in\mom(k,\lambda,\sigma)}
 \Pp_P\!\left\{|\widehat\mu-\E_PX|>\eps\right\}\leq\delta
\]
using $O_k(\mathfrak r_k)$ samples.  This is the formal version of the COLT
open problem stated above.

Definition~\ref{def:nonadaptive} is a public-coin specialization of the model
in \citet[Section~2]{lau2026open}, which permits arbitrary measurable query
sets and allows the final decoder to use the completed transcript.  Thus a
decoded center may determine decoder weights, provided that it does not
alter the query sets.  Proposition~\ref{prop:fixed-borel} gives an explicit
Borel representation of the sets used by both constructions.

\subsection{Summary of contributions}

\citet[Theorem~5]{lau2026order} attain
$O_k(\mathfrak r_k)$ samples with adaptive randomized threshold queries.
They also give an order-optimal two-stage general-query protocol
\citep[Section~4.3]{lau2026order}: a fixed coding-theoretic first stage
localizes the mean to an $O(\sigma)$ interval, after which a second refinement
query block is generated using the decoded interval.  Thus one adaptive
transition was sufficient.  Their non-adaptive lower bound rules out
threshold and interval queries, but not arbitrary measurable query sets;
the arbitrary-query case remained unresolved.  Separately, their
Theorem~9 gives a matching $\Omega_k(\mathfrak r_k)$ lower bound for arbitrary
1-bit protocols, even adaptive ones, in its stated small-error,
high-confidence range.

Theorem~\ref{thm:main} shows that the last adaptive transition is unnecessary.
For every fixed $k>1$, a fully non-adaptive public-coin protocol uses
$O_k(\mathfrak r_k)$ samples for every $0<\delta<1/2$.  Both its localization
and refinement queries are fixed before any message is observed; the decoded
center changes only the decoder's weights on the pre-existing refinement
bits.  Hence zero adaptive transitions suffice in the model of arbitrary
measurable queries.  Combined with the preceding lower bound, the protocol is
minimax optimal in the parameter range formally covered by that lower bound.

Beyond the rate itself, the paper makes three technical contributions.
\begin{itemize}
  \item We construct a \emph{universal one-batch refinement}: every query is
  generated before communication, yet the decoder can later specialize the
  same stored bits to whichever coarse interval was localized.
  \item We introduce safe periodic residues and an adjacent-scale telescope.
  The resulting correction at scale $L$ vanishes unless a sample crosses a
  grid boundary, yielding moment-sensitive rather than range-sensitive
  variance.
  \item We give two complementary refinement constructions.  The primary
  construction uses dyadic periodic residues and decoder-side correlated
  thresholds.  A second construction, developed in
  Section~\ref{sec:continuous-overview}, uses continuous random scales,
  shifted grids, Rademacher cell colorings, and an adjacent-cell gradient.
  It reaches the same rates without the dyadic telescope.
\end{itemize}

The main technical obstacle is universal refinement.  The fixed coding block
already solves coarse localization, but location-dependent refinement queries
are easy to design only after an $O(\sigma)$ interval is known.  A fully
non-adaptive protocol must instead serve every possible interval
simultaneously.  Periodic quantizers reuse a bit at many locations, but under
a mere moment assumption a far-tail sample can alias into the local
coordinate system and destroy the variance bound.  We remove this obstacle
by moving all dependence on the coarse location from the quantizer to the
decoder.  Safe periodic phases, decoder-side correlated thresholding, and an
adjacent-scale telescope confine aliasing to a terminal tail term and make
the scale-$L$ variance occur only when $|X-c|\gtrsim L$.

\subsection{Proof ideas}

The proof combines four ingredients.

First, decoder-side correlated thresholding allows a fixed bit to be
recentered after acquisition.  For a bounded function $h\in[a,a+R]$, draw a public
$U\sim\Unif[a,a+R]$ and transmit $B=\1\{U\leq h(X)\}$.  If the decoder later
learns a reference point $c$, it can compute $B_c=\1\{U\leq h(c)\}$ without
another message.  The random variable
\[
 R(B-B_c)
\]
is unbiased for $h(X)-h(c)$ and has conditional second moment
$R|h(X)-h(c)|$.  This is the scalar correlated-sampling idea behind
Wyner--Ziv estimators \citep{mayekar2021wyner}.

Second, safe periodic residues provide an adjacent-scale telescope.  At scale
$L_j=2^jL_0$, we prepare two grids shifted by $L_j/2$.  Whatever
$c$ is, one grid puts it at distance at least $L_j/4$ from every boundary.
Let $r_j(x)$ be the residue of $x$ in the safe grid and define
$f_j=r_{j+1}-r_j$.  Although the safe phases are unknown when the queries are
fixed, there are only four phase pairs for $f_j$; randomizing uniformly over
them and importance weighting at the decoder costs only a constant.  The
identity
\[
 [r_0(X)-r_0(c)]
 +\sum_{j=0}^{J-1}[f_j(X)-f_j(c)]
 =r_J(X)-r_J(c)
\]
concentrates all periodic aliasing at the terminal scale.

Third, the same safe phases yield a crossing-supported variance bound.  If
$|X-c|<L_j/4$, both safe residues change by exactly $X-c$, so
$f_j(X)-f_j(c)=0$.  Hence the scale-$j$ second moment is at most
\[
 C L_j^2\Pp\{|X-c|\geq L_j/4\}.
\]
We sample only one scale per device.  The importance distribution
\[
 p_j\ \propto\ L_j^{(2-k)/2}
\]
is exactly the square-root allocation for this moment envelope.  It
concentrates on small scales when $k>2$, is uniform when $k=2$, and
concentrates on large scales when $1<k<2$.  This single choice produces the
three regimes in \eqref{eq:colt-rate}.

Finally, the second construction replaces the dyadic telescope by a
continuum of random grid widths.  A Rademacher coloring makes the decoder
sensitive only to the two cells adjacent to the cell containing $c$.  At
width $r$, the resulting expectation is a compactly supported function of
$|X-c|/r$.  Integrating this kernel over all $r>0$ reproduces $X-c$ exactly;
finite cutoffs control the small-displacement and tail bias.  Sampling one
scale from a moment-matched density then gives the same three rates.  The
construction is stated in Section~\ref{sec:continuous-overview}, with a
complete proof in Appendix~\ref{app:continuous}.

\subsection{Related work}

\paragraph{Unquantized mean estimation under weak moments.}
The statistical benchmark in \eqref{eq:colt-rate} comes from mean estimation
under minimal moment assumptions.  Median-type procedures and other robust
estimators give sub-Gaussian confidence under finite variance
\citep{devroye2016subgaussian,lee2022optimal,minsker2023efficient}.  Under only
a $(1+\alpha)$-moment, matching possibility and impossibility results
characterize the attainable rate
\citep{devroye2016subgaussian,cherapanamjeri2022optimal}.  The latter work also
gives information-theoretic and computationally efficient estimators without
a finite variance; recent work studies finer distribution-dependent notions
of optimality \citep{dang2023optimality}.  These works observe the samples
directly.  Our question is whether their one-dimensional accuracy exponents
survive when each sample is replaced by one preassigned bit.

\paragraph{Distributed estimation under information constraints.}
Bandwidth-constrained estimation has a long history in decentralized sensing
\citep{luo2005universal,ribeiro2006gaussian,ribeiro2006unknown}.
Information-theoretic lower bounds were subsequently developed through
communication complexity, distributed data processing, geometric
representations, and Fisher information
\citep{NIPS2013_d6ef5f7f,shamir2014fundamental,
braverman2016communication,han2018geometric,barnes2020fisher,
duchi2019locallyprivate,acharya2022interactive,acharya2023unified}.
These frameworks cover high-dimensional parametric or structured families
and several interactive communication models.  Sharp 1-bit mean estimators
often exploit a known distributional shape, including Gaussian location,
symmetric log-concave families, or a location--scale model
\citep{kipnis2022mean,cai2024distributed,kumar2026unknown}.  Such parametric
structure can support efficient threshold or Gray-code quantizers, but it is
absent from $\mom(k,\lambda,\sigma)$.

\paragraph{When interaction helps.}
Interaction creates exponential or polynomial improvements for certain
distributed learning, sparse estimation, and locally private
hypothesis-selection problems
\citep{dagan2020interaction,acharya2022role,gopi2020locally,pour2024sample}.
Conversely, it does not improve the order for every task; for example,
\citet{kazemi2025sample} establish no sequential-interaction gain for
distributed simple binary hypothesis testing.  Thus interactivity is a
property of the statistical task together with its admissible query class,
not of the communication budget alone.  The present result gives another
sharp separation: for scalar finite-moment mean estimation, threshold and
interval restrictions retain an interaction gap, whereas arbitrary
measurable 1-bit queries admit a fully non-adaptive optimum.

\paragraph{Finite-moment 1-bit estimation.}
\citet{lau2026sequential} introduced near-optimal sequential interval-query
estimators under a variance bound and exhibited a two-stage general-query
variant.  \citet{lau2026order} extend the theory to all fixed central moments
$k>1$, establishing the rate $\mathfrak r_k$ in \eqref{eq:colt-rate}.
The fully non-adaptive arbitrary-query case was then isolated by
\citet{lau2026open}.  Our improvement is in interaction complexity, not in
the minimax sample rate or query simplicity: their main protocol uses
adaptive threshold queries, whereas our fully non-adaptive refinement uses
nonlocal periodic measurable sets.

The one-sample statistical query model of \citet[Definition~1.3 and
Theorem~2.1]{feldman2017range} permits one arbitrary Boolean query per fresh
sample.  Its multiscale mean estimators highlight the cost of an unknown
range, while its variance-sensitive guarantee first finds a center
\citep[Lemmas~3.1--3.2 and Theorem~3.8]{feldman2017range}.  Our decoder-side
construction can be viewed as a way to perform that centering after the bits
have already been sent.  Correlated threshold sampling is closely related to
the public-threshold identity in
\citet[Section~4.1]{mayekar2021wyner}.  Non-contiguous periodic
cells also occur in universal scalar quantization
\citep{boufounos2012universal}, but there the objective is geometric signal
recovery rather than distribution-free finite-moment estimation.  We are not
aware of prior work combining decoder-side centering, periodic telescoping,
and moment-matched scale sampling to obtain the rates above.

Concurrent work by \citet{abdalla2026robust} studies robust mean estimation
using dithered 1-bit quantization.  Its distributed construction assumes a
quantization range or rough quantile information covering the unknown
location, while its location-free variant begins with unquantized samples.
Here every sample is instead represented by a single bit from a measurable query
fixed before communication, with no prior bound on the location.

\section{Setup and main theorem}
\label{sec:setup}

Let $X_1,\ldots,X_n\stackrel{\mathrm{iid}}\sim P$.  A randomized 1-bit
protocol chooses measurable query sets $A_1,\ldots,A_n\subseteq\R$ and
observes
\[
 Y_t=\1\{X_t\in A_t\}.
\]
The protocol is \emph{fully non-adaptive} if all query sets, including their
public random seeds, are chosen before any $Y_t$ is observed.  The final
decoder may use the full transcript and all public randomness.  In
particular, reinterpreting a fixed bit after a coarse estimate has been
decoded is allowed; changing its query set is not.

An estimator $\widehat\mu$ is $(\eps,\delta)$-accurate over a class
$\mathcal P$ if
\[
 \sup_{P\in\mathcal P}
 \Pp_P\{|\widehat\mu-\E_PX|>\eps\}\leq\delta.
\]
Here and below, probability is joint over the i.i.d. samples and all public
randomness used to generate the query sets.

\begin{theorem}[Fully non-adaptive order-optimal mean estimation]
\label{thm:main}
Fix $k>1$.  There are constants $c_k,C_k>0$ such that for every
$\lambda\geq\sigma>0$, $0<\eps\leq c_k\sigma$, and
$0<\delta<1/2$, there is a fully non-adaptive public-coin 1-bit protocol
that is $(\eps,\delta)$-accurate over $\mom(k,\lambda,\sigma)$ and uses
\[
 n\leq C_k\ell(\lambda/\sigma)
 +
 \begin{cases}
 C_k\dfrac{\sigma^2}{\eps^2}\log\dfrac{2}{\delta}, & k>2,\\[2mm]
 C_k\dfrac{\sigma^2}{\eps^2}\ell(\sigma/\eps)
       \log\dfrac{2}{\delta}, & k=2,\\[2mm]
 C_k(\dfrac{\sigma}{\eps})^{\frac{k}{k-1}}
       \log\dfrac{2}{\delta}, & 1<k<2.
 \end{cases}
\]
\end{theorem}

We use the following existing localization result.

\begin{proposition}[Non-adaptive localization]
\label{prop:localization}
Theorem~16 of \citet[version~2]{lau2026order} gives the following.
There is a universal constant $C_{\loc}$ and a deterministic fully
non-adaptive 1-bit protocol using
\[
 O\!\left(\ell(\lambda/\sigma)+\log\frac1\eta\right)
\]
samples such that, whenever $|\E X|\leq\lambda$ and
$\E|X-\E X|\leq\sigma$, it returns an interval $I$ of length at most
$C_{\loc}\sigma$ containing $\E X$ with probability at least $1-\eta$.
\end{proposition}

The first-moment premise is automatically satisfied on the class studied
here.  Indeed, Lyapunov's inequality gives, for every
$P\in\mom(k,\lambda,\sigma)$,
\begin{equation}
\label{eq:lyapunov-localization}
 \E|X-\mu|
 \leq\bigl(\E|X-\mu|^k\bigr)^{1/k}
 \leq\sigma.
\end{equation}

The query block in Proposition~\ref{prop:localization} is a fixed global
error-correcting code.  We run it on samples independent of all refinement
samples.  Let $c$ be the midpoint of the returned interval.  On the
localization-success event,
\begin{equation}
\label{eq:center-distance}
 |c-\mu|\leq C_{\loc}\sigma/2,
 \qquad \mu=\E X.
\end{equation}
All localization and refinement queries, including mutually independent
public coins for the refinement samples, are generated together before either
block is observed.  The center $c$ is a function only of the localization
transcript and never of the refinement samples or their public seeds.

\begin{lemma}[Conditional decoupling of refinement]
\label{lem:conditional-decoupling}
Let $\mathcal F_{\loc}$ be the sigma-field generated by the localization
transcript, and let $c$ be its decoded midpoint.  For any refinement block
used by either construction, write $S_i$ for all public seeds assigned to
device $i$.  Conditional on $\mathcal F_{\loc}$, the pairs $(X_i,S_i)$ retain
their unconditional product law.  Consequently, for any fixed decoder map
$T$, statistics of the form $T(X_i,S_i;c)$ are conditionally i.i.d.\ within
a refinement block.
This applies in particular to the dyadic statistics $W_{0,i}(c)$ and
$W_i(c)$ and to the continuous statistic $Z_{c,i}$ in
\eqref{eq:continuous-main-statistic}.  On the localization-success event,
all moment bounds below hold uniformly for the realized $c$ satisfying
\eqref{eq:center-distance}.
\end{lemma}

The proof is given in Appendix~\ref{app:dyadic-conditioning}.  This
separation is what permits the rest of the analysis to fix an arbitrary
successful center $c$ without turning the protocol into a two-stage one.

For later use, define the deterministic moment-inflation factor and localized
scale
\begin{equation}
\label{eq:tau-definition}
 \Gamma_k=
 \left[2^{k-1}\left(1+(C_{\loc}/2)^k\right)\right]^{1/k},
 \qquad \tau=\Gamma_k\sigma.
\end{equation}

\section{A universal non-adaptive refinement protocol}
\label{sec:protocol}

Throughout this section, condition on an arbitrary center $c$ satisfying
\eqref{eq:center-distance}.  Lemma~\ref{lem:conditional-decoupling} makes
this conditioning legitimate and preserves the product structure within
each refinement block.

\subsection{Periodic coordinates with safe decoder phases}

For $L>0$ and $b\in\{0,1\}$, define the shifted residue
\begin{equation}
\label{eq:residue}
 \rho_{L,b}(x)
 =x-\frac{bL}{2}
 -L\left\lfloor\frac{x-bL/2}{L}\right\rfloor
 \in[0,L).
\end{equation}
Its discontinuities form the grid
$\mathcal G_{L,b}=bL/2+L\mathbb Z$.

\begin{lemma}[Safe phase]
\label{lem:safe-phase}
For every $c\in\R$ and $L>0$, at least one $b\in\{0,1\}$ satisfies
\[
 \dist(c,\mathcal G_{L,b})\geq L/4.
\]
Consequently, if this $b$ is used and $|x-c|<L/4$, then
\[
 \rho_{L,b}(x)-\rho_{L,b}(c)=x-c.
\]
\end{lemma}

\begin{proof}
The two boundary grids are shifted by half a period.  A point at distance
less than $L/4$ from one grid is at distance greater than $L/4$ from the
other, with equality only at a tie.  For a safe grid, the interval joining
$c$ and any $x$ with $|x-c|<L/4$ contains no boundary, so the residue has
slope one throughout that interval.
\end{proof}

Figure~\ref{fig:safe-phase} depicts this boundary-free neighborhood for the
selected half-period shift.

\begin{figure}[t]
\centering
\resizebox{\columnwidth}{!}{\input{fig_main_mechanism.tex}}
\caption{Safe phase at one scale.  The boundary grids
$\mathcal G_{L,0}$ and $\mathcal G_{L,1}$ differ by half a period.  At least
one phase has no boundary in the $L/4$-neighborhood of $c$ (gray), so its
residue is linear throughout that neighborhood.}
\label{fig:safe-phase}
\end{figure}
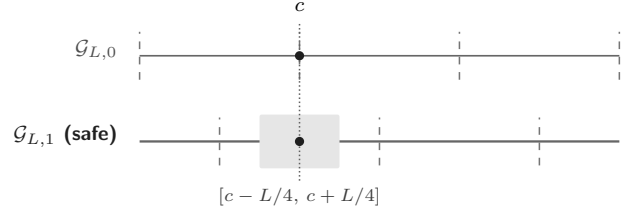

Fix a deterministic tie rule and write $b_L(c)$ for a safe phase.  Set
\begin{equation}
\label{eq:scales}
 \begin{aligned}
 L_0&=8\tau, & L_j&=2^jL_0,\\
 b_j&=b_{L_j}(c), & r_j&=\rho_{L_j,b_j}.
 \end{aligned}
\end{equation}
For every phase pair $(b,b')\in\{0,1\}^2$, define
\begin{equation}
\label{eq:fj-pair}
 f_{j,b,b'}(x)
 =\rho_{L_{j+1},b'}(x)-\rho_{L_j,b}(x)
 \in(-L_j,2L_j).
\end{equation}
After $c$ is decoded, the selected function is
$f_j=f_{j,b_j,b_{j+1}}=r_{j+1}-r_j$.

\subsection{The 1-bit decoder-side identity}

\begin{lemma}[Correlated threshold bit]
\label{lem:correlated-bit}
Let $h:\R\to[a,a+R]$ be measurable, let
$U\sim\Unif[a,a+R]$, and let the encoder transmit
$B_x=\1\{U\leq h(x)\}$.\footnote{The choice of $\leq$ rather than $<$ is
immaterial because $U$ is continuous and independent of $X$; this remains
true when $P$ has atoms.  We use the same convention for all public
thresholds below.}  A decoder knowing $y$ and $U$ can form
\[
 Z=R\big(B_x-\1\{U\leq h(y)\}\big).
\]
For every fixed $x,y$,
\begin{equation}
\label{eq:correlated-identities}
 \E_U Z=h(x)-h(y),
 \qquad
 \E_U Z^2=R|h(x)-h(y)|.
\end{equation}
\end{lemma}

\begin{proof}
For any $u,v\in[a,a+R]$,
\[
 \E_U[\1\{U\leq u\}-\1\{U\leq v\}]=(u-v)/R.
\]
The two indicators differ exactly when $U$ lies between $u$ and $v$, an
event of probability $|u-v|/R$.  Multiplying by $R$ and $R^2$ gives the two
claims.
\end{proof}

\subsection{The base-coordinate block}

Each base query independently draws $B\sim\Unif\{0,1\}$ and
$U\sim\Unif[0,L_0]$, then transmits the single bit
\begin{equation}
\label{eq:base-query}
 Y=\1\{U\leq\rho_{L_0,B}(X)\}.
\end{equation}
After localization, the decoder forms
\begin{equation}
\label{eq:base-stat}
 W_0
 =2\1\{B=b_0\}L_0
 \left(Y-\1\{U\leq\rho_{L_0,B}(c)\}\right).
\end{equation}

By Lemma~\ref{lem:correlated-bit} and the random phase,
\begin{equation}
\label{eq:base-moments}
 \E W_0=\E[r_0(X)-r_0(c)],
 \qquad
 \E W_0^2\leq2L_0^2.
\end{equation}

\subsection{The randomized-scale correction block}

Let $J\geq1$ be a terminal scale to be fixed in
Section~\ref{sec:analysis}.  Define for
$j\in\{0,\ldots,J-1\}$
\begin{equation}
\label{eq:scale-probabilities}
 w_j=2^{j(2-k)/2},
 \qquad
 p_j=\frac{w_j}{\sum_{i=0}^{J-1}w_i}.
\end{equation}
Thus $p_j$ decreases geometrically for $k>2$, is uniform for $k=2$, and
increases geometrically for $1<k<2$.

Each correction query independently draws
\[
 \begin{aligned}
 K&\sim(p_0,\ldots,p_{J-1}),\\
 (B,B')&\sim\Unif\{0,1\}^2,\qquad
 U\sim\Unif[-L_K,2L_K],
 \end{aligned}
\]
and transmits a single bit
\begin{equation}
\label{eq:correction-query}
 Y=\1\{U\leq f_{K,B,B'}(X)\}.
\end{equation}
After $c$ is decoded, set
\begin{equation}
\label{eq:correction-stat}
 \begin{aligned}
 W={}&\frac{4}{p_K}\1\{(B,B')=(b_K,b_{K+1})\}\,3L_K\\
 &\times\left(Y-\1\{U\leq f_{K,B,B'}(c)\}\right).
 \end{aligned}
\end{equation}
The phase indicator in \eqref{eq:correction-stat} discards the three pairs
that are unsafe for the decoded center.

\begin{lemma}[Correction mean and second moment]
\label{lem:correction-moments}
Conditionally on $c$,
\begin{align}
 \E W
 &=\sum_{j=0}^{J-1}\E[f_j(X)-f_j(c)],
 \label{eq:corr-mean}\\
 \E W^2
 &\leq36\sum_{j=0}^{J-1}
 \frac{L_j^2}{p_j}
 \Pp\{|X-c|\geq L_j/4\}.
 \label{eq:corr-second}
\end{align}
\end{lemma}

The first identity is the adjacent-scale telescope in expectation.  The
second is the crossing-supported variance bound: at scale $L_j$, a
nonzero correction requires $X$ to leave the boundary-free neighborhood of
$c$.  The importance-weight and phase calculations are given in
Appendix~\ref{app:dyadic-correction}.

\subsection{Protocol summary}
\label{sec:protocol-summary}

The dyadic protocol has three stages.  Its localization, base, and
correction blocks use disjoint fresh samples.

\begin{enumerate}
\item[\textbf{1.}] \textbf{Public query compilation.}
Given $(k,\lambda,\sigma,\eps,\delta)$, set
$\tau,L_0,\ldots,L_J$ and $(p_j)_{j<J}$ according to
\eqref{eq:tau-definition}, \eqref{eq:scales},
\eqref{eq:terminal-J}, and \eqref{eq:scale-probabilities}, and allocate the
three sample blocks according to Proposition~\ref{prop:localization},
\eqref{eq:base-cost}, and \eqref{eq:corr-cost}.  Before observing any
message, fix all localization queries.  For every base sample $i$,
independently draw $B_i\sim\Unif\{0,1\}$ and
$U_i\sim\Unif[0,L_0]$, then fix
$Q_i^{(0)}(x)=\1\{U_i\leq\rho_{L_0,B_i}(x)\}$.  For every correction sample
$i$, independently draw $K_i\sim(p_0,\ldots,p_{J-1})$,
$(B_i,B_i')\sim\Unif\{0,1\}^2$, and
$U_i\sim\Unif[-L_{K_i},2L_{K_i}]$, then fix
$Q_i^{(1)}(x)=\1\{U_i\leq f_{K_i,B_i,B_i'}(x)\}$.

\item[\textbf{2.}] \textbf{One-shot acquisition.}
Apply every fixed query to its assigned sample and collect the resulting
1-bit messages.  The three blocks may be acquired in parallel.

\item[\textbf{3.}] \textbf{Decoder-only reconstruction.}
Decode the localization interval $I$, let $c$ be its midpoint, and choose
$b_j=b_{L_j}(c)$ for $j=0,\ldots,J$.  Convert the stored refinement bits
into $W_{0,i}$ and $W_i$ using \eqref{eq:base-stat} and
\eqref{eq:correction-stat}, assigning zero weight to unsafe phase pairs.
Let $\widehat m_0$ and $\widehat m_1$ be the corresponding median-of-means
estimates and return
\[
 \widehat\mu=c+\widehat m_0+\widehat m_1.
\]
\end{enumerate}

No query in Steps 1--2 depends on the localization transcript; the decoded
center enters only in Step 3.

\section{Analysis of the dyadic protocol}
\label{sec:analysis}

\subsection{Proof roadmap}

On the event \eqref{eq:center-distance}, for $D=X-c$,
\begin{equation}
\label{eq:centered-moment}
 \E|D|^k
 \leq2^{k-1}\left(\E|X-\mu|^k+|\mu-c|^k\right)
 \leq\Gamma_k^k\sigma^k,
\end{equation}
where $\Gamma_k$ and $\tau$ are defined in
\eqref{eq:tau-definition}.  Thus $\E|D|^k\leq\tau^k$, and
\eqref{eq:scales} sets $L_0=8\tau$ directly.

The proof follows the decoder-only reconstruction in the third stage above.
The base block estimates a periodic coordinate at scale $L_0$.  The
correction block adds its changes over adjacent scales.  These means telescope
to a single terminal residue, whose discrepancy from $D$ is a tail event.
Finally, the scale law in \eqref{eq:scale-probabilities} balances the second
moments of the correction terms.  A complete proof, including the
concentration argument, is given in Appendix~\ref{app:dyadic-proof}.

\subsection{Bias: telescope to a terminal scale}

Define
\begin{equation}
\label{eq:terminal-J}
 \begin{aligned}
 C_k^{\mathrm{tail}}&=5\cdot4^{k-1},\\
 J&=\min\left\{j\geq0:
 C_k^{\mathrm{tail}}\frac{\tau^k}{L_j^{k-1}}
 \leq\frac{\eps}{4}\right\}.
 \end{aligned}
\end{equation}
We take the constant $c_k$ in Theorem~\ref{thm:main} small enough that
\begin{equation}
\label{eq:ck-choice}
 c_k<\min\left\{1,
 \frac{4C_k^{\mathrm{tail}}\Gamma_k}{8^{k-1}}
 \right\}.
\end{equation}
This choice makes $J\geq1$; the verification and the precise use of
minimality are deferred to Appendix~\ref{app:dyadic-bias}.

\begin{lemma}[All aliasing is terminal]
\label{lem:terminal-bias}
The expectations targeted by the base and correction blocks telescope as
\begin{equation}
\label{eq:telescope}
 \begin{aligned}
 &\E[r_0(X)-r_0(c)]
 +\sum_{j=0}^{J-1}\E[f_j(X)-f_j(c)]\\
 &\hspace{22mm}=\E[r_J(X)-r_J(c)].
 \end{aligned}
\end{equation}
Moreover,
\begin{equation}
\label{eq:bias-bound}
 \left|\E[r_J(X)-r_J(c)]-(\mu-c)\right|
 \leq C_k^{\mathrm{tail}}\frac{\tau^k}{L_J^{k-1}}
 \leq\eps/4.
\end{equation}
\end{lemma}

The identity \eqref{eq:telescope} is visible directly from
$f_j=r_{j+1}-r_j$.  Its statistical role is more important than its
algebraic simplicity: every intermediate alias cancels, so only a terminal
sample with $|D|\geq L_J/4$ can contribute bias.  Minimality in
\eqref{eq:terminal-J} then gives
\begin{equation}
\label{eq:LJ-size}
 L_J\asymp_k
 \sigma\left(\frac{\sigma}{\eps}\right)^{1/(k-1)},
 \qquad
 J=O_k(\ell(\sigma/\eps)).
\end{equation}
The two-sided inequality behind \eqref{eq:LJ-size}, including all constants,
appears in Appendix~\ref{app:dyadic-bias}.

\subsection{Variance: moment-matched scale allocation}

\begin{lemma}[Moment-matched scale allocation]
\label{lem:variance-regimes}
Under \eqref{eq:centered-moment}, the statistic $W$ in
\eqref{eq:correction-stat} satisfies
\begin{equation}
\label{eq:variance-regimes}
 \E W^2\leq
 \begin{cases}
 C_k\sigma^2, & k>2,\\
 C\sigma^2 J, & k=2,\\
 C_k\sigma^kL_J^{2-k}, & 1<k<2.
 \end{cases}
\end{equation}
\end{lemma}

For $k\neq2$, the moment condition gives the scale envelope
\[
 L_j^2\Pp\{|D|\geq L_j/4\}
 \lesssim_k\tau^kL_j^{2-k}.
\]
The variance-minimizing importance distribution is proportional to the
square root of this envelope, which is precisely
$p_j\propto L_j^{(2-k)/2}$.  If
$S_J=\sum_{j<J}2^{j(2-k)/2}$, this yields the compact calculation
\[
 \E W^2\lesssim_k\tau^kL_0^{2-k}S_J^2.
\]
The sum $S_J$ is bounded for $k>2$ and is of order
$2^{J(2-k)/2}$ for $1<k<2$, producing the first and third lines of
\eqref{eq:variance-regimes}.  At the boundary $k=2$, the allocation is
uniform and the sharper pointwise bound
\begin{equation}
\label{eq:pointwise-dyadic}
 \sum_{j=0}^{J-1}L_j^2\1\{|d|\geq L_j/4\}
 \leq\frac{64}{3}d^2
\end{equation}
gives the middle line.  This pointwise summation is essential: a separate
Markov bound at every scale would introduce an artificial second logarithm.
Appendix~\ref{app:dyadic-variance} gives the details.

\begin{remark}[$k$-dependence and the transition at two]
\label{rem:k-transition}
The theorem treats $k$ as fixed, and the constants hidden by $O_k(\cdot)$
are not uniform as $k\to2$.  To see the source explicitly, put
$a=|k-2|/2$.  The geometric normalizer in
the preceding display is
\[
 S_J=
 \begin{cases}
 \dfrac{1-2^{-aJ}}{1-2^{-a}}, & k>2,\\[2mm]
 \dfrac{2^{aJ}-1}{2^a-1}, & 1<k<2.
 \end{cases}
\]
Thus the constants obtained from the fixed-$k$ geometric bounds grow at
most quadratically in $1/|k-2|$ near the boundary.  For every fixed $J$,
$S_J\to J$ as $k\to2$; exactly at $k=2$, however, the pointwise summation
\eqref{eq:pointwise-dyadic} improves the resulting $J^2$ bound to $J$.
This is the origin of the critical $\ell(\sigma/\eps)$ factor.  A theorem
uniform over moment orders $k=k(\eps)$ approaching two would require a
separate interpolation analysis and is not claimed here.
\end{remark}

\subsection{From second moments to the theorem}

Run localization with failure probability $\delta/3$, and use independent
median-of-means estimates $\widehat m_0$ and $\widehat m_1$ for the base and
correction blocks, each with target deviation $\eps/4$ and failure
probability $\delta/3$.  The decoder returns
\begin{equation}
\label{eq:final-estimator}
 \widehat\mu=c+\widehat m_0+\widehat m_1.
\end{equation}
By \eqref{eq:base-moments}, the base sample cost is
\begin{equation}
\label{eq:base-cost}
 n_0=O_k\!\left(\frac{\sigma^2}{\eps^2}
 \log\frac{1}{\delta}\right)
\end{equation}
and Lemma~\ref{lem:variance-regimes} gives
\begin{equation}
\label{eq:corr-cost}
 n_1\leq
 \begin{cases}
 O_k(\frac{\sigma^2}{\eps^2}\log\frac1\delta), & k>2,\\[1mm]
 O(\frac{\sigma^2}{\eps^2}J\log\frac1\delta), & k=2,\\[1mm]
 O_k(\frac{\sigma^kL_J^{2-k}}{\eps^2}\log\frac1\delta),
 & 1<k<2.
 \end{cases}
\end{equation}
Substituting \eqref{eq:LJ-size} into the heavy-tail line gives
\[
 \frac{\sigma^kL_J^{2-k}}{\eps^2}
 =O_k\!\left((\sigma/\eps)^{k/(k-1)}\right),
\]
while $J=O(\ell(\sigma/\eps))$ gives the critical rate.  The terminal bias
and the two estimation errors use only $3\eps/4$ of the error budget.
The localization confidence term is absorbed by the base cost when
$\eps\leq c_k\sigma$.  These observations prove the claimed orders; the
detailed conditional median-of-means argument, union bound, integer
rounding, and non-adaptivity check are in
Appendix~\ref{app:dyadic-concentration}.

\begin{corollary}[Resolution of the interaction question]
\label{cor:optimality}
For every fixed $k>1$, there are constants $c'_k,\delta_0>0$ such that, when
$0<\eps<c'_k\sigma$ and $0<\delta<\delta_0$, the minimax sample complexity of
fully non-adaptive arbitrary 1-bit mean estimation over
$\mom(k,\lambda,\sigma)$ is, up to $k$-dependent constants, the rate in
Theorem~\ref{thm:main}.
\end{corollary}

\begin{proof}
Let $c_k^{\rm up}$ be the small-error constant in
Theorem~\ref{thm:main}, and let $c_k^{\rm lb}$ and $\delta_0$ be the
constants in the matching lower bound of
\citet[Theorem~9]{lau2026order}.  Take
$c'_k=\min\{c_k^{\rm up},c_k^{\rm lb}\}$.  That lower bound holds for
arbitrary 1-bit estimators, including adaptive protocols, and contains
both the additive $\Omega(\log(\lambda/\sigma))$ localization term and the
three refinement terms in \eqref{eq:colt-rate}.  It therefore applies to
the fully non-adaptive subclass and matches Theorem~\ref{thm:main} throughout
the claimed parameter range.
\end{proof}

\section{A complementary continuous-scale refinement}
\label{sec:continuous-overview}

The dyadic protocol uses finitely many scales and an adjacent-scale
telescope.  The construction below retains the same localization and
concentration primitives, but replaces periodic residues by a random-grid
identity over a continuum of scales.  It attains the same upper bound through
a reconstruction mechanism that is not tied to dyadic grids.  A complete
derivation is given in Appendix~\ref{app:continuous}.

Condition on a successful localization transcript and put
$D=X-c$, so that $\E|D|^k\leq\tau^k$ by
\eqref{eq:centered-moment}.  Fix $a=1/4$.  Before observing any bit, each
refinement device independently draws a public scale $R$ with density $p$, a
shift $U\mid R\sim\Unif[0,R)$, and Rademacher cell colors
$(\xi_m)_{m\in\mathbb Z}$.  With
\[
 q_{r,u}(x)=\left\lfloor\frac{x+u}{r}\right\rfloor,
\]
the device sends
\[
 B=\1\{\xi_{q_{R,U}(X)}=+1\},
 \qquad \widetilde B=2B-1.
\]
Only after the center has been decoded, the decoder sets
\[
\begin{aligned}
 Q&=q_{R,U}(c),&
 V&=\frac{c+U}{R}-Q,\\
 A&=\1\{a\leq V\leq1-a\}.&&
\end{aligned}
\]
and evaluates
\begin{equation}
\label{eq:continuous-main-statistic}
 Z_c=\frac{A}{C_ap(R)}
 (\widetilde B-\xi_Q)(\xi_{Q+1}-\xi_{Q-1}),
 \qquad C_a=\log\frac{15}{7}.
\end{equation}
The encoder does not use $c$.  The gate, control variate, and adjacent-cell
gradient in \eqref{eq:continuous-main-statistic} are decoder-side operations
on an already transmitted bit.

Define the compactly supported kernel
\begin{equation}
\label{eq:continuous-main-kernel}
 \psi_a(t)
 =\operatorname{Leb}\bigl([a,1-a]\cap[1-t,2-t)\bigr),
 \qquad t\geq0.
\end{equation}
Its support is $[a,2-a]$ and
\begin{equation}
\label{eq:continuous-main-normalization}
 \int_0^\infty\frac{\psi_a(t)}{t^2}\,dt=C_a.
\end{equation}
Consequently, for every $d\in\R$,
\begin{equation}
\label{eq:continuous-main-reproduction}
 \frac{\sgn(d)}{C_a}
 \int_0^\infty\psi_a\!\left(\frac{|d|}{r}\right)dr=d.
\end{equation}
Equation~\eqref{eq:continuous-main-reproduction} is the continuous analogue
of the dyadic telescope: both identities reconstruct displacement while
allowing the center to enter only at the decoder.

Figure~\ref{fig:random-grid-kernel} illustrates the decoder-side cancellation
and the compact kernel that underlies this reproduction identity.

\begin{figure}[t]
\centering
\resizebox{\columnwidth}{!}{\input{fig_random_grid_kernel.tex}}
\caption{Continuous-scale refinement.  (a)~With $c$ in the central half of
cell $Q$, color averaging gives responses $-1$, $0$, and $+1$ on the left
adjacent, central, and right adjacent cells, respectively, and zero on all
other cells.  (b)~Shift averaging yields the compact kernel
$\psi_{1/4}(|d|/r)$, whose normalized all-scale integral in
\eqref{eq:continuous-main-reproduction} equals $d=X-c$.}
\label{fig:random-grid-kernel}
\end{figure}
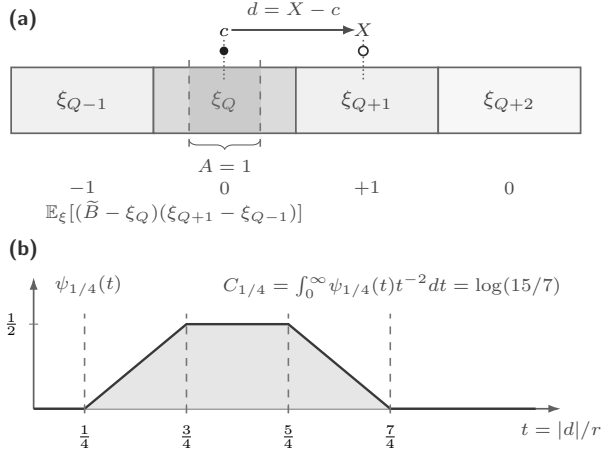

\begin{proposition}[Continuous-scale refinement]
\label{prop:continuous-refinement}
Suppose $0<\eps\leq\tau$ and define
\begin{equation}
\label{eq:continuous-main-cutoffs}
 r_-=\frac{\eps}{14},
 \qquad
 r_+=4\left(\frac{8\tau^k}{\eps}\right)^{1/(k-1)}.
\end{equation}
There is a density $p$ on $[r_-,r_+]$, depending only on
$(k,\tau,\eps)$, for which \eqref{eq:continuous-main-statistic} satisfies
\begin{align}
\label{eq:continuous-main-bias}
 \left|\E Z_c-(\mu-c)\right|&\leq\frac{\eps}{4},\\
\label{eq:continuous-main-second}
 \E Z_c^2&\leq
 \begin{cases}
  C_k\tau^2, &k>2,\\
  C\tau^2\ell(\tau/\eps), &k=2,\\
  C_k\tau^2(\tau/\eps)^{\frac{2-k}{k-1}}, &1<k<2.
 \end{cases}
\end{align}
Moreover, conditional on its public seed, every encoder bit is the indicator
of a fixed Borel subset of $\R$ independent of the localization transcript.
Thus median-of-means applied to independent copies of $Z_c$ gives the
refinement terms in Theorem~\ref{thm:main} with a fully non-adaptive query
block.
\end{proposition}

\begin{proof}[Proof sketch]
For integers $m,j$, pairwise independence and centering of the colors give
\[
\begin{aligned}
 &\E_\xi[(\xi_m-\xi_j)(\xi_{j+1}-\xi_{j-1})]\\
 &\qquad=\1\{m=j+1\}-\1\{m=j-1\}.
\end{aligned}
\]
Conditionally on $R=r$ and $D=d$, averaging the shift makes the event that
$X$ lies in the appropriate adjacent cell have probability
$\psi_a(|d|/r)$.  The factor $1/p(R)$ then cancels the scale density, giving
\begin{equation}
\label{eq:continuous-main-truncated}
 \E[Z_c\mid D=d]
 =\frac{\sgn(d)}{C_a}
 \int_{r_-}^{r_+}\psi_a\!\left(\frac{|d|}{r}\right)dr.
\end{equation}
The kernel contributes only when
$|d|/(2-a)<r<|d|/a$.  Hence the truncated integral equals $d$ for
\[
 \frac{\eps}{8}\leq|d|\leq
 \left(\frac{8\tau^k}{\eps}\right)^{1/(k-1)}.
\]
Outside this window, truncation preserves the sign and cannot increase the
magnitude.  The small-displacement contribution is at most $\eps/8$, and
the moment bound controls the upper tail by another $\eps/8$, proving
\eqref{eq:continuous-main-bias}.

The same-cell control variate in
\eqref{eq:continuous-main-statistic} vanishes unless $X$ crosses a grid
boundary.  Since the gate places $c$ at distance at least $aR$ from either
boundary,
\begin{equation}
\label{eq:continuous-main-envelope}
 \E[Z_c^2\mid D=d]
 \leq\frac{16}{C_a^2}
 \int_{r_-}^{r_+}
 \1\{r\leq|d|/a\}\frac{dr}{p(r)}.
\end{equation}
Choose $p(r)\propto r^{1-k}$ for $1<k<2$,
$p(r)\propto r^{-1}$ for $k=2$, and, for $k>2$,
\[
 p(r)\propto
 \begin{cases}
  \tau^{-1}, &r\leq\tau,\\
  \tau^{k-2}r^{1-k}, &r>\tau.
 \end{cases}
\]
Substitution in \eqref{eq:continuous-main-envelope}, followed by
$\E|D|^k\leq\tau^k$, yields \eqref{eq:continuous-main-second}.  Dividing
these bounds by $\eps^2$ and applying Lemma~\ref{lem:mom} gives respectively
$(\tau/\eps)^2$, $(\tau/\eps)^2\ell(\tau/\eps)$, and
$(\tau/\eps)^{k/(k-1)}$ samples, up to the common
$\log(1/\delta)$ factor.  The cutoff and moment calculations are completed
in Appendix~\ref{app:continuous}.
\end{proof}

\begin{remark}[Comparison of the two constructions]
Both constructions attain the same statistical rate through different
reconstruction identities.  The dyadic route uses a finite scale
distribution and a direct telescope, whereas the continuous route uses one
refinement block and an all-scale reproduction identity.  Thus non-adaptive
refinement is not specific to the dyadic representation.
\end{remark}

\section{Conclusion}

We resolved the interaction question for 1-bit mean estimation with
arbitrary measurable queries: a fully non-adaptive protocol attains the same
order-optimal finite-moment rates as the best adaptive procedures.  Both
constructions combine a fixed localization block with universal refinement
queries.  The dyadic construction uses safe periodic residues and an
adjacent-scale telescope, whereas the continuous construction uses shifted
random grids, cell colorings, and a compactly supported integral identity.

The common principle is to separate query adaptivity from decoder
adaptivity.  Every refinement bit is acquired before the coarse location is
known; the decoded center changes only how the stored bits are interpreted.
This separation permits a single query plan to serve every candidate
location, while boundary-crossing identities provide the moment-sensitive
variance needed for the three optimal tail regimes.  The two constructions
show that this principle is not tied to one algebraic representation of the
refinement.

Whether such universal refinements extend to higher dimensions, privacy
constraints, or restricted query geometries remains open.

\bibliographystyle{plainnat}
\bibliography{references}

\clearpage
\appendix
\onecolumn
\thispagestyle{plain}

\begin{center}
{\LARGE\bfseries Appendix\par}
\vspace{0.5em}
{\large Universal Refinement without Interaction:\\
Order-Optimal 1-Bit Mean Estimation\par}
\end{center}
\vspace{0.5em}
\hrule
\vspace{1.25em}

\section*{Overview}
The appendices provide complete proofs and supporting numerical validation.
Appendix~\ref{app:dyadic-proof} establishes Theorem~\ref{thm:main};
Appendix~\ref{app:continuous} establishes
Proposition~\ref{prop:continuous-refinement};
Appendix~\ref{app:implementation} reports the numerical experiments; and
Appendix~\ref{app:query-sets} verifies that the randomized quantizers are
Borel measurable and fixed before communication.

\section{Proof of Theorem~\ref{thm:main}: the dyadic protocol}
\label{app:dyadic-proof}

Conditioning on successful localization fixes the decoded center while
preserving independence across the refinement blocks.  The proof then
establishes the base and correction identities, controls the terminal bias,
derives the three variance bounds, and concludes with concentration.

\subsection{Step 1: localization and conditional independence}
\label{app:dyadic-conditioning}

\begin{proof}[Proof of Lemma~\ref{lem:conditional-decoupling}]
The localization block and every refinement block use disjoint i.i.d.\
samples.  Every refinement seed is generated independently of all samples,
of the localization seeds, and of every other refinement seed.  Therefore
the collection $(X_i,S_i)_i$ is independent of
$\mathcal F_{\loc}$.  Conditioning on $\mathcal F_{\loc}$ leaves its
product law unchanged and makes the decoded midpoint $c$ deterministic.
The moment bounds use $c$ only through
\eqref{eq:center-distance}; hence they hold uniformly over every successful
localization transcript.
\end{proof}

On a successful localization transcript,
$D=X-c=(X-\mu)+(\mu-c)$ and the convexity inequality
$|u+v|^k\leq2^{k-1}(|u|^k+|v|^k)$ give
\[
 \E|D|^k
 \leq2^{k-1}\bigl(\E|X-\mu|^k+|\mu-c|^k\bigr)
 \leq2^{k-1}\left(1+(C_{\loc}/2)^k\right)\sigma^k
 =\tau^k.
\]
This proves \eqref{eq:centered-moment}.  All expectations below are
conditional on such a fixed successful transcript unless stated otherwise.

\subsection{Step 2: the base coordinate}
\label{app:dyadic-base}

We verify \eqref{eq:base-moments}.  Conditional on $B=b_0$, apply
Lemma~\ref{lem:correlated-bit} with range length $L_0$ and reference $c$.
Because $\Pp\{B=b_0\}=1/2$, the decoder weight in
\eqref{eq:base-stat} gives
\[
 \E W_0
 =2\cdot\frac12\,
 \E[r_0(X)-r_0(c)]
 =\E[r_0(X)-r_0(c)].
\]
For the second moment, the same identity yields
\[
 \E W_0^2
 =4\cdot\frac12\,L_0
   \E|r_0(X)-r_0(c)|
 \leq2L_0^2,
\]
because both residues lie in $[0,L_0)$.  Lemma~\ref{lem:conditional-decoupling}
justifies this calculation conditional on the decoded center.

\subsection{Step 3: adjacent-scale corrections}
\label{app:dyadic-correction}

\begin{proof}[Proof of Lemma~\ref{lem:correction-moments}]
Condition on $K=j$.  If the random phase pair is the selected pair, apply
Lemma~\ref{lem:correlated-bit} to
$f_j\in(-L_j,2L_j)$, an interval of length $3L_j$.  Averaging the factor
$4/p_j$ over the scale probability $p_j$ and phase-pair probability $1/4$
gives
\[
 \E W
 =\sum_{j=0}^{J-1}\E[f_j(X)-f_j(c)],
\]
which is \eqref{eq:corr-mean}.

For the second moment, Lemma~\ref{lem:correlated-bit} gives
\[
 \begin{aligned}
 &\E_U\!\left[
 \big\{3L_j(Y-\1\{U\leq f_j(c)\})\big\}^2
 \middle|X,c,K=j,(B,B')=(b_j,b_{j+1})\right]\\
 &\hspace{35mm}=3L_j|f_j(X)-f_j(c)|.
 \end{aligned}
\]
If $|X-c|<L_j/4$, Lemma~\ref{lem:safe-phase} applies at both
$L_j$ and $L_{j+1}$.  Each selected residue then changes by $X-c$, so
$f_j(X)-f_j(c)=0$.  In every case,
$|f_j(X)-f_j(c)|\leq3L_j$, and the preceding display is bounded by
\[
 9L_j^2\1\{|X-c|\geq L_j/4\}.
\]
Squaring the importance weight and averaging over the scale and phase pair
contributes
$p_j(1/4)(4/p_j)^2=4/p_j$.  Summing over $j$ proves
\eqref{eq:corr-second}.
\end{proof}

\subsection{Step 4: telescoping and terminal bias}
\label{app:dyadic-bias}

\begin{proof}[Proof of Lemma~\ref{lem:terminal-bias}]
Since $f_j=r_{j+1}-r_j$, the pointwise sum telescopes:
\[
 [r_0(X)-r_0(c)]
 +\sum_{j=0}^{J-1}[f_j(X)-f_j(c)]
 =r_J(X)-r_J(c).
\]
Taking expectations proves \eqref{eq:telescope}.

Safety at scale $L_J$ implies
$r_J(X)-r_J(c)=D$ whenever $|D|<L_J/4$.  On the complementary event,
both residues belong to $[0,L_J)$, and hence
\[
 |D-(r_J(X)-r_J(c))|
 \leq |D|+L_J
 \leq5|D|,
\]
where the last inequality uses $|D|\geq L_J/4$.  Therefore
\begin{align*}
 \left|\E[D-(r_J(X)-r_J(c))]\right|
 &\leq5\E\!\left[|D|\1\{|D|\geq L_J/4\}\right]\\
 &\leq5\frac{\E|D|^k}{(L_J/4)^{k-1}}
 \leq C_k^{\mathrm{tail}}\frac{\tau^k}{L_J^{k-1}}.
\end{align*}
The definition \eqref{eq:terminal-J} bounds the last expression by
$\eps/4$, proving \eqref{eq:bias-bound}.
\end{proof}

The terminal-scale bounds follow from the criterion in
\eqref{eq:terminal-J} and its minimality.  At $j=0$, the left side equals
\[
 C_k^{\mathrm{tail}}\frac{\tau^k}{L_0^{k-1}}
 =\frac{C_k^{\mathrm{tail}}\tau}{8^{k-1}}.
\]
Since $\eps\leq c_k\sigma=c_k\tau/\Gamma_k$, the strict choice
\eqref{eq:ck-choice} makes this quantity larger than $\eps/4$; thus
$J\geq1$.  The criterion holds at $J$ and, by minimality, fails at $J-1$.
Consequently,
\begin{equation}
\label{eq:LJ-two-sided}
 4C_k^{\mathrm{tail}}\frac{\tau^k}{\eps}
 \leq L_J^{k-1}
 <2^{k-1}\,4C_k^{\mathrm{tail}}\frac{\tau^k}{\eps}.
\end{equation}
Indeed, the left inequality is the criterion at $J$, while the failed
criterion at $J-1$ gives the right inequality after using
$L_J=2L_{J-1}$.  Taking $(k-1)$-st roots and substituting
$\tau=\Gamma_k\sigma$ proves \eqref{eq:LJ-size}.  Finally,
$J=\log_2(L_J/L_0)$ with $L_0=8\tau$ gives
$J=O_k(\ell(\sigma/\eps))$.

\subsection{Step 5: the three variance regimes}
\label{app:dyadic-variance}

\begin{proof}[Proof of Lemma~\ref{lem:variance-regimes}]
First suppose $k\neq2$.  Markov's inequality and
\eqref{eq:centered-moment} give, at every scale,
\begin{equation}
\label{eq:scale-envelope}
 L_j^2\Pp\{|D|\geq L_j/4\}
 \leq4^k\tau^kL_j^{2-k}.
\end{equation}
Let $S_J=\sum_{i=0}^{J-1}2^{i(2-k)/2}$.  Substituting
$p_j=2^{j(2-k)/2}/S_J$ into \eqref{eq:corr-second} and using
$L_j=2^jL_0$ yields
\begin{align}
 \E W^2
 &\leq C_k\tau^kL_0^{2-k}
 S_J\sum_{j=0}^{J-1}2^{j(2-k)/2}
 \notag\\
 &=C_k\tau^kL_0^{2-k}S_J^2.
 \label{eq:variance-geometric}
\end{align}
If $k>2$, then $S_J\leq(1-2^{(2-k)/2})^{-1}=O_k(1)$.
Since $L_0=8\tau$, \eqref{eq:variance-geometric} is
$O_k(\tau^2)=O_k(\sigma^2)$.  If $1<k<2$, then
\[
 S_J
 =\frac{2^{J(2-k)/2}-1}{2^{(2-k)/2}-1}
 =O_k(2^{J(2-k)/2}),
\]
so \eqref{eq:variance-geometric} is
$O_k(\tau^kL_J^{2-k})=O_k(\sigma^kL_J^{2-k})$.

For $k=2$, $p_j=1/J$.  Fix $d\in\R$, and let $j_\star$ be the largest
included index in the left side of \eqref{eq:pointwise-dyadic}, if one
exists.  Then $L_{j_\star}\leq4|d|$ and
\[
 \sum_{j=0}^{j_\star}L_j^2
 =L_{j_\star}^2\sum_{h=0}^{j_\star}4^{-h}
 \leq\frac43L_{j_\star}^2
 \leq\frac{64}{3}d^2.
\]
If there is no included index, the sum is zero.  This proves
\eqref{eq:pointwise-dyadic}.  Taking expectations and using
$\E D^2\leq\tau^2$, then applying \eqref{eq:corr-second}, gives
\[
 \E W^2
 \leq36J\sum_{j=0}^{J-1}
 L_j^2\Pp\{|D|\geq L_j/4\}
 \leq768J\tau^2.
\]
This proves all three cases.
\end{proof}

\subsection{Step 6: concentration and completion of the proof}
\label{app:dyadic-concentration}
\label{app:mom}

\begin{lemma}[Median-of-means concentration]
\label{lem:mom}
Let $Z_1,\ldots,Z_n$ be independent copies of a random variable with mean
$m$ and $\E Z_1^2\leq V$.  There is a universal constant $C$ such that, for
every $t>0$ and $0<\eta<1/2$, if
\[
 n\geq C\max\left\{1,\frac{V}{t^2}\right\}\log\frac1\eta,
\]
then the following estimator differs from $m$ by at most $t$ with
probability at least $1-\eta$: split the observations into
$q=\lceil8\log(1/\eta)\rceil$ equal-sized blocks, discard at most $q-1$
leftover observations, and take the median of the $q$ block means.
\end{lemma}

\begin{proof}
Split the observations into $q=\lceil8\log(1/\eta)\rceil$ blocks of size
$s=\lfloor n/q\rfloor$.  For a sufficiently large universal $C$, the
sample-size assumption ensures
$s\geq16\max\{1,V/t^2\}$.  For a block mean $\overline Z$, Chebyshev's
inequality and $\Var(Z_1)\leq\E Z_1^2\leq V$ give
\[
 \Pp\{|\overline Z-m|>t\}
 \leq \frac{\Var(Z_1)}{st^2}
 \leq\frac{V}{st^2}
 \leq\frac1{16}.
\]
The block-failure indicators are independent.  The median fails only if at
least half of the blocks fail, whose probability is at most
\[
 \sum_{h=\lceil q/2\rceil}^{q}\binom qh(1/16)^h
 \leq2^q(1/16)^{q/2}=e^{-q\log2}\leq\eta.
\]
Rounding the block sizes and discarding leftovers change only the universal
sample-size constant.
\end{proof}

\begin{proof}[Proof of Theorem~\ref{thm:main}]
Run Proposition~\ref{prop:localization} with failure probability
$\delta/3$.  Use an independent base block whose median-of-means estimate
$\widehat m_0$ targets $\E W_0$ to accuracy $\eps/4$ with failure
probability $\delta/3$, and an independent correction block whose estimate
$\widehat m_1$ targets $\E W$ to the same accuracy and failure probability.
Return \eqref{eq:final-estimator}.

Condition on $\mathcal F_{\loc}$.  On every successful localization
transcript, Lemma~\ref{lem:conditional-decoupling} fixes $c$ while preserving
the product laws in both refinement blocks.  The second-moment bounds
\eqref{eq:base-moments} and \eqref{eq:variance-regimes} therefore hold
uniformly for that transcript, so Lemma~\ref{lem:mom} may be applied
conditionally.  It gives the base cost \eqref{eq:base-cost} and the
correction costs \eqref{eq:corr-cost}.

For $1<k<2$, \eqref{eq:LJ-size} yields
\[
 \frac{\sigma^kL_J^{2-k}}{\eps^2}
 \leq C_k
 \frac{\sigma^k}{\eps^2}
 \left(\frac{\sigma^k}{\eps}\right)^{\frac{2-k}{k-1}}
 =C_k\left(\frac{\sigma}{\eps}\right)^{\frac{k}{k-1}},
\]
where the exponents satisfy
$k+k(2-k)/(k-1)=k/(k-1)$ and
$2+(2-k)/(k-1)=k/(k-1)$.  For $k=2$,
$J=O(\ell(\sigma/\eps))$.  For $1<k<2$, the exponent
$k/(k-1)>2$, so the base cost is absorbed by the correction cost; for
$k\geq2$, absorption follows directly from $J\geq1$ and the first two
lines of \eqref{eq:corr-cost}.

On the intersection of localization success and the two median-of-means
success events, Lemmas~\ref{lem:correction-moments} and
\ref{lem:terminal-bias} give
\begin{align*}
 |\widehat\mu-\mu|
 &\leq|\widehat m_0-\E W_0|
      +|\widehat m_1-\E W|\\
 &\quad+
 \left|\E[r_J(X)-r_J(c)]-(\mu-c)\right|\\
 &\leq\eps/4+\eps/4+\eps/4<\eps.
\end{align*}
The conditional failure probabilities are uniform over all successful
transcripts.  Averaging them over $\mathcal F_{\loc}$ and taking a union
bound with localization therefore gives total failure probability at most
$\delta$.

The localization block costs
$O(\ell(\lambda/\sigma)+\log(3/\delta))$ samples.  Since
$\eps\leq c_k\sigma$, its confidence term is absorbed by
\eqref{eq:base-cost}; ceilings and discarded median-of-means observations
change only universal constants.  Combining the three block sizes yields
the sample complexity in Theorem~\ref{thm:main}.

Finally, we verify non-adaptivity.  The localization code is
deterministic.  Each base query is determined by $(B,U)$, and each
correction query by $(K,B,B',U)$, all sampled before any message is
observed.  The decoded center appears only in the subtracted reference
indicators, safe-phase selection, and importance weights at the decoder.
Thus no encoder query depends on the transcript, and the protocol is fully
non-adaptive.
\end{proof}

\section{Proof of the continuous-scale refinement}
\label{app:continuous}

This appendix proves Proposition~\ref{prop:continuous-refinement}.  A device
sends a random projection of a random-binning feature, and the decoder probes
an adjacent-cell discrete gradient at the localized center.  Integrating the
resulting compactly supported kernel over a continuum of spatial scales
reproduces displacement; finite cutoffs and moment-matched scale sampling
give the three rates in Theorem~\ref{thm:main}.

By Lemma~\ref{lem:conditional-decoupling}, we may condition on the
localization transcript.  On a successful transcript, the center $c$ is
fixed, the refinement samples and their public coins retain their product
law, and
\begin{equation}
\label{eq:alt-local-moment}
 D=X-c,
 \qquad
 \E|D|^k\leq\tau^k,
 \qquad
 \tau=\Gamma_k\sigma,
\end{equation}
by \eqref{eq:centered-moment}.  Every public coin used below is generated
before any localization or refinement bit is observed.

\subsection{A signed random-bin primitive}

Fix
\[
 a=\frac14.
\]
For $t\geq0$, define the compactly supported kernel
\begin{equation}
\label{eq:alt-kernel}
 \psi_a(t)
 =\operatorname{Leb}\bigl([a,1-a]\cap[1-t,2-t)\bigr).
\end{equation}
An interval calculation gives
\begin{equation}
\label{eq:alt-kernel-pieces}
 \psi_a(t)=
 \begin{cases}
  0, &0\leq t\leq a,\\
  t-a, &a<t<1-a,\\
  1-2a, &1-a\leq t\leq1+a,\\
  2-a-t, &1+a<t<2-a,\\
  0, &t\geq2-a.
 \end{cases}
\end{equation}
Its normalization is
\begin{align}
 C_a
 &=\int_0^\infty\frac{\psi_a(t)}{t^2}\,dt
 =\int_a^{1-a}\int_{1-v}^{2-v}\frac{dt}{t^2}\,dv
 \notag\\
 &=\log\!\left(\frac{(1-a)(1+a)}{a(2-a)}\right)
 =\log\!\left(\frac{15}{7}\right).
\label{eq:alt-Ca}
\end{align}

Let $0<r_-<r_+$, and let $p$ be a positive probability density on
$[r_-,r_+]$.  Independently for each refinement device, draw the public
randomness
\begin{equation}
\label{eq:alt-public-coins}
 R\sim p,
 \qquad U\mid R\sim\Unif[0,R),
 \qquad (\xi_m)_{m\in\mathbb Z}
 \stackrel{\mathrm{iid}}{\sim}\Unif\{-1,+1\}.
\end{equation}
The coloring is independent of $(R,U,X)$ and of the localization
sigma-field, and the collections in \eqref{eq:alt-public-coins}
are independent across refinement devices.
For $r>0$ and $u\in[0,r)$, write
\[
 q_{r,u}(x)=\left\lfloor\frac{x+u}{r}\right\rfloor.
\]
The encoder observes $X$ and sends a single bit
\begin{equation}
\label{eq:alt-encoder-bit}
 B=\1\{\xi_{q_{R,U}(X)}=+1\},
 \qquad \widetilde B=2B-1.
\end{equation}
After the center $c$ has been decoded, set
\begin{equation}
\label{eq:alt-decoder-cell}
 Q=q_{R,U}(c),
 \qquad
 V=\frac{c+U}{R}-Q\in[0,1),
 \qquad
 A=\1\{a\leq V\leq1-a\}.
\end{equation}
The decoder forms
\begin{equation}
\label{eq:alt-statistic}
 Z_c
 =\frac{A}{C_ap(R)}
   (\widetilde B-\xi_Q)(\xi_{Q+1}-\xi_{Q-1}).
\end{equation}
The first difference is a same-cell control variate: it vanishes identically
when $X$ and $c$ occupy the same cell.  The second difference is an oriented
adjacent-cell gradient.  Figure~\ref{fig:random-grid-kernel} depicts this
cancellation pattern and the resulting compact kernel.

\begin{lemma}[Continuous random-grid identity]
\label{lem:alt-grid}
For every $d\in\R$,
\begin{equation}
\label{eq:alt-eta}
 \E[Z_c\mid D=d]
 =\eta_{r_-,r_+}(d)
 :=\frac{\sgn(d)}{C_a}
   \int_{r_-}^{r_+}\psi_a\!\left(\frac{|d|}{r}\right)dr.
\end{equation}
The function $\eta_{r_-,r_+}$ has the following properties:
\begin{align}
 \eta_{r_-,r_+}(d)&=d
 &&\text{if }(2-a)r_-\leq|d|\leq ar_+,
 \label{eq:alt-exact-window}\\
 d\,\eta_{r_-,r_+}(d)&\geq0,
 \qquad |\eta_{r_-,r_+}(d)|\leq|d|
 &&\text{for every }d.
 \label{eq:alt-sign-bound}
\end{align}
Moreover,
\begin{equation}
\label{eq:alt-second-master}
 \E[Z_c^2\mid D=d]
 \leq\frac{16}{C_a^2}
 \int_{r_-}^{r_+}
 \1\left\{r\leq\frac{|d|}{a}\right\}\frac{dr}{p(r)}.
\end{equation}
\end{lemma}

\begin{proof}
For arbitrary integers $m$ and $j$, independence and centering of the
Rademacher signs give
\begin{equation}
\label{eq:alt-rademacher-identity}
 \E_\xi[(\xi_m-\xi_j)(\xi_{j+1}-\xi_{j-1})]
 =\1\{m=j+1\}-\1\{m=j-1\}.
\end{equation}
Condition on $R=r$ and $D=d>0$.  The fractional coordinate $V$ in
\eqref{eq:alt-decoder-cell} is uniform on $[0,1)$.  With $t=d/r$, the
sample lies in the right-adjacent cell exactly when
$1\leq V+t<2$.  Intersecting this event with the gate
$V\in[a,1-a]$ gives probability $\psi_a(t)$.  For $d<0$, reflection about
the center of $[0,1)$ gives the same probability for the left-adjacent
cell.  Combining this observation with
\eqref{eq:alt-rademacher-identity}, then averaging over $R$, cancels
$p(R)$ and proves \eqref{eq:alt-eta}.

For $d\neq0$, the substitution $t=|d|/r$ and
\eqref{eq:alt-Ca} give the full-scale identity
\begin{equation}
\label{eq:alt-full-reproduction}
 \frac{\sgn(d)}{C_a}
 \int_0^\infty\psi_a\!\left(\frac{|d|}{r}\right)dr
 =\frac{d}{C_a}\int_0^\infty\frac{\psi_a(t)}{t^2}\,dt
 =d.
\end{equation}
The integrand is nonzero only for
\[
 \frac{|d|}{2-a}<r<\frac{|d|}{a}.
\]
This support is contained in $[r_-,r_+]$ under the condition in
\eqref{eq:alt-exact-window}, which proves exact reproduction there.
Truncating the nonnegative magnitude integral in
\eqref{eq:alt-full-reproduction} preserves its weak direction and can only
decrease its magnitude, which proves \eqref{eq:alt-sign-bound}.

For the second moment, each difference in \eqref{eq:alt-statistic} has
absolute value at most two.  If $q_{R,U}(X)=Q$, the first difference is
zero.  On $A=1$, the center is at distance at least $aR$ from both cell
boundaries, so
\[
 A=1,\quad q_{R,U}(X)\neq Q
 \quad\Longrightarrow\quad |D|\geq aR.
\]
Consequently,
\[
 Z_c^2
 \leq\frac{16}{C_a^2p(R)^2}
 \1\left\{R\leq\frac{|D|}{a}\right\}.
\]
Averaging over $R$ proves \eqref{eq:alt-second-master}.
\end{proof}

For every realization of $(R,U,\xi)$, the encoder query in
\eqref{eq:alt-encoder-bit} is the Borel set
\begin{equation}
\label{eq:alt-query-set}
 \bigcup_{m:\,\xi_m=+1}[mR-U,(m+1)R-U).
\end{equation}
The set definition is independent of $c$.  The infinite public coloring is valid in
the arbitrary measurable-query model; only pairwise independence of its
coordinates is needed for \eqref{eq:alt-rademacher-identity}.

\subsection{Finite scale cutoffs}

Assume $0<\eps\leq\tau$, and choose
\begin{equation}
\label{eq:alt-cutoffs}
 r_-=\frac{\eps}{8(2-a)}=\frac{\eps}{14},
 \qquad
 r_+=\frac1a\left(\frac{8\tau^k}{\eps}\right)^{1/(k-1)}.
\end{equation}
Thus the exact window in \eqref{eq:alt-exact-window} is
\begin{equation}
\label{eq:alt-explicit-window}
 \frac{\eps}{8}\leq|d|\leq
 \left(\frac{8\tau^k}{\eps}\right)^{1/(k-1)}.
\end{equation}

\begin{lemma}[Continuous-scale truncation bias]
\label{lem:alt-bias}
Under \eqref{eq:alt-local-moment} and \eqref{eq:alt-cutoffs},
\begin{equation}
\label{eq:alt-bias}
 \left|\E Z_c-\E D\right|\leq\frac{\eps}{4}.
\end{equation}
\end{lemma}

\begin{proof}
Put
\[
 L=(2-a)r_-=\frac\eps8,
 \qquad
 H=ar_+=\left(\frac{8\tau^k}{\eps}\right)^{1/(k-1)}.
\]
Lemma~\ref{lem:alt-grid} gives exact reproduction when
$L\leq|D|\leq H$, and \eqref{eq:alt-sign-bound} gives
$|\eta_{r_-,r_+}(D)-D|\leq|D|$ elsewhere.  Hence
\begin{align*}
 \left|\E Z_c-\E D\right|
 &\leq L+\E[|D|\1\{|D|>H\}]\\
 &\leq\frac\eps8+\frac{\E|D|^k}{H^{k-1}}
 \leq\frac\eps8+\frac\eps8.
\end{align*}
\end{proof}

\subsection{Moment-matched continuous scale sampling}

The density $p$ now depends on the moment regime.  All three choices use
only the known parameters $(k,\sigma,\eps)$ through $\tau$, $r_-$, and
$r_+$; none depends on $c$ or the localization transcript.

\begin{lemma}[Continuous scale allocation]
\label{lem:alt-scales}
There is a density $p$ on $[r_-,r_+]$ for which the statistic
\eqref{eq:alt-statistic} satisfies
\begin{equation}
\label{eq:alt-three-variances}
 \E Z_c^2\leq
 \begin{cases}
  C_k\tau^2, &k>2,\\
  C\tau^2\log(e\tau/\eps), &k=2,\\
  C_k\tau^kr_+^{2-k}, &1<k<2.
 \end{cases}
\end{equation}
Consequently,
\begin{equation}
\label{eq:alt-three-rates}
 \frac{\E Z_c^2}{\eps^2}\leq
 \begin{cases}
  C_k(\tau/\eps)^2, &k>2,\\
  C(\tau/\eps)^2\log(e\tau/\eps), &k=2,\\
  C_k(\tau/\eps)^{k/(k-1)}, &1<k<2.
 \end{cases}
\end{equation}
\end{lemma}

\begin{proof}
Suppose first that $1<k<2$.  Let
\[
 N_< =\int_{r_-}^{r_+}r^{1-k}\,dr
 =\frac{r_+^{2-k}-r_-^{2-k}}{2-k},
 \qquad
 p(r)=\frac{r^{1-k}}{N_<}.
\]
Applying \eqref{eq:alt-second-master} and then integrating only up to
$|D|/a$ gives
\begin{align*}
 \E Z_c^2
 &\leq\frac{16N_<}{C_a^2}
   \E\int_{r_-}^{r_+}
   \1\left\{r\leq\frac{|D|}{a}\right\}r^{k-1}\,dr\\
 &\leq\frac{16N_<}{kC_a^2a^k}\E|D|^k
 \leq C_k\tau^kr_+^{2-k}.
\end{align*}
Substitution of \eqref{eq:alt-cutoffs} gives the last line of
\eqref{eq:alt-three-rates}.

For $k=2$, take the log-uniform density
\[
 N_2=\log\frac{r_+}{r_-},
 \qquad
 p(r)=\frac{1}{N_2r}.
\]
Equation~\eqref{eq:alt-second-master} yields
\[
 \E Z_c^2
 \leq\frac{8N_2}{C_a^2a^2}\E D^2
 \leq C N_2\tau^2.
\]
Here
\[
 \frac{r_+}{r_-}=448\left(\frac\tau\eps\right)^2,
\]
so $N_2=O(\log(e\tau/\eps))$.

Finally, suppose $k>2$.  Define
\[
 q(r)=
 \begin{cases}
  \tau^{-1}, &r_-\leq r\leq\tau,\\
  \tau^{k-2}r^{1-k}, &\tau<r\leq r_+,
 \end{cases}
 \qquad
 N_>=\int_{r_-}^{r_+}q(r)\,dr,
 \qquad p(r)=\frac{q(r)}{N_>}.
\]
The assumption $\eps\leq\tau$ ensures $r_-<\tau<r_+$, and
\begin{equation}
\label{eq:alt-light-normalizer}
 \frac{13}{14}\leq N_>
 \leq1+\frac{1}{k-2}.
\end{equation}
For every $u\geq0$,
\[
 \int_{r_-}^{r_+}\1\{r\leq u\}\frac{dr}{p(r)}
 \leq N_>\left[
  \tau\min\{u,\tau\}+\frac{\tau^{2-k}u^k}{k}
 \right].
\]
Use $u=|D|/a$ in \eqref{eq:alt-second-master}.  Lyapunov's inequality
and \eqref{eq:alt-local-moment} give $\E|D|\leq\tau$ and
$\E|D|^k\leq\tau^k$, and therefore
\[
 \E Z_c^2
 \leq\frac{16N_>\tau^2}{C_a^2}
 \left(\frac1a+\frac{1}{ka^k}\right)
 \leq C_k\tau^2.
\]
This proves \eqref{eq:alt-three-variances}; dividing by $\eps^2$ proves
the first two lines of \eqref{eq:alt-three-rates}, while the heavy-tail
algebra was established in the first case.
\end{proof}

\subsection{Alternative proof of the main theorem}

Generate independently in advance one random scale, shift, and cell coloring
for each refinement sample, using the appropriate density from
Lemma~\ref{lem:alt-scales}.  After all bits have arrived and the localization
center $c$ has been decoded, form independent copies
$Z_{c,1},\ldots,Z_{c,n}$ of \eqref{eq:alt-statistic}.  Conditional on the
full localization transcript, Lemmas~\ref{lem:alt-bias} and
\ref{lem:alt-scales} apply uniformly on the localization-success event.

Apply Lemma~\ref{lem:mom} with target deviation $\eps/2$ and failure
probability $\delta/2$, and let $\widehat m$ be the resulting median of
means.  With the refinement sample sizes in \eqref{eq:alt-three-rates}
multiplied by $C\log(2/\delta)$,
\[
 |\widehat m-\E Z_c|\leq\frac\eps2
\]
with conditional probability at least $1-\delta/2$.  On this event,
\[
 |c+\widehat m-\mu|
 \leq|\widehat m-\E Z_c|
    +|\E Z_c-\E(X-c)|
 \leq\frac{3\eps}{4}<\eps.
\]
The localization block fails with probability at most $\delta/2$; averaging
the conditional guarantee over successful transcripts and taking a union
bound gives total failure probability at most $\delta$.

The resulting sample complexity is exactly the order in
Theorem~\ref{thm:main}, because $\tau=\Gamma_k\sigma$.  Every encoder query
is fixed by $(R,U,\xi)$ before communication and is independent of $c$.
Only the decoder-side gate, control variate, and adjacent-cell gradient use
$c$.  This proves Proposition~\ref{prop:continuous-refinement}; together
with the shared localization block, it gives a second proof of the upper
bound in Theorem~\ref{thm:main}.

\section{Numerical experiments and reproducibility}
\label{app:implementation}

\subsection{Experimental setup}

We evaluate implementations of both refinement constructions.  The dyadic
implementation has separate base and correction blocks over $J$ discrete
scales.  The continuous implementation uses one refinement block, samples a
scale in $[r_-,r_+]$, and realizes the cell colors with a
pairwise-independent affine family.  Table~\ref{tab:implementations} summarizes the
structural differences.

\begin{table}[t]
\centering
\small
\setlength{\tabcolsep}{6pt}
\setlength{\belowcaptionskip}{6pt}
\renewcommand{\arraystretch}{1.22}
\caption{Structural comparison of the two refinement implementations.}
\label{tab:implementations}
\begin{tabular}{lcc}
\toprule
Property & Dyadic & Continuous\\
\midrule
Scale support & $J$ levels & $[r_-,r_+]$\\
Refinement blocks & base $+$ correction & one\\
Public variables & $K,B,B',U$ & $R,U,\xi$\\
Cell coloring & none & pairwise-independent\\
\bottomrule
\end{tabular}
\end{table}

We set $\sigma=1$, $\mu=0.2$, and supply the decoder with $c=0$, using
$\tau=[2^{k-1}(1+0.2^k)]^{1/k}$.  For $k=3$ we use a centered Gaussian
rescaled to have third absolute central moment one.  For $k=2$ and $k=1.5$
we use symmetric Pareto laws with tail indices $2.3$ and $1.7$, respectively,
again rescaled to have $k$-th central moment one.  At each of five target
ratios we draw $400{,}000$ independent observations.

\subsection{Analytical evaluation of public randomness}

Direct simulation of the importance-weighted 1-bit statistics has high
Monte Carlo variance, especially when $k<2$.  We therefore integrate the
public randomness analytically and sample only $X$.  For the dyadic
construction, let
\[
 \Delta_j(x;c)=r_j(x)-r_j(c),\qquad 0\leq j\leq J,
\]
where each $r_j$ uses the safe phase selected by $c$.  Lemma~\ref{lem:correlated-bit}
and the phase probabilities give the conditional identities
\begin{align}
 \E[W_0\mid X=x]&=\Delta_0(x;c),
 &\E[W_0^2\mid X=x]&=2L_0|\Delta_0(x;c)|,\label{eq:numerical-dyadic-base}\\
 \E[W\mid X=x]&=\sum_{j=0}^{J-1}
   \bigl(\Delta_{j+1}(x;c)-\Delta_j(x;c)\bigr),
 &\E[W^2\mid X=x]&=12\sum_{j=0}^{J-1}
   \frac{L_j}{p_j}|\Delta_{j+1}(x;c)-\Delta_j(x;c)|.
 \label{eq:numerical-dyadic-correction}
\end{align}
Thus the integrated conditional mean is exactly $\Delta_J(x;c)$ and the
second moment is the sum of the two right-hand expressions.

For the continuous construction, put $D=x-c$, $a=1/4$, and
\[
 \chi_a(t)=
 \begin{cases}
  0,&0\leq t\leq a,\\
  t-a,&a<t<1-a,\\
  1-2a,&t\geq1-a.
 \end{cases}
\]
Averaging first over the shift and the pairwise-independent cell signs yields
the exact conditional second moment
\begin{equation}
 \E[Z_c^2\mid X=x]
 =\frac{4}{C_a^2}\int_{r_-}^{r_+}
   \frac{\chi_a(|D|/r)}{p(r)}\,dr.
 \label{eq:numerical-continuous-second}
\end{equation}
We evaluate this piecewise integral and the corresponding conditional-mean
integral in closed form.  Literal 1-bit simulations provide a separate
comparison with these Rao--Blackwellized values.

\subsection{Rate validation}

Define the confidence-free refinement rate
\begin{equation}
\label{eq:numerical-rate}
 \mathfrak v_k(\tau,\eps)=
 \begin{cases}
  (\tau/\eps)^2, &k>2,\\
  (\tau/\eps)^2\log(e\tau/\eps), &k=2,\\
  (\tau/\eps)^{k/(k-1)}, &1<k<2.
 \end{cases}
\end{equation}

\begin{figure}[t]
\centering
\includegraphics[width=.85\linewidth]{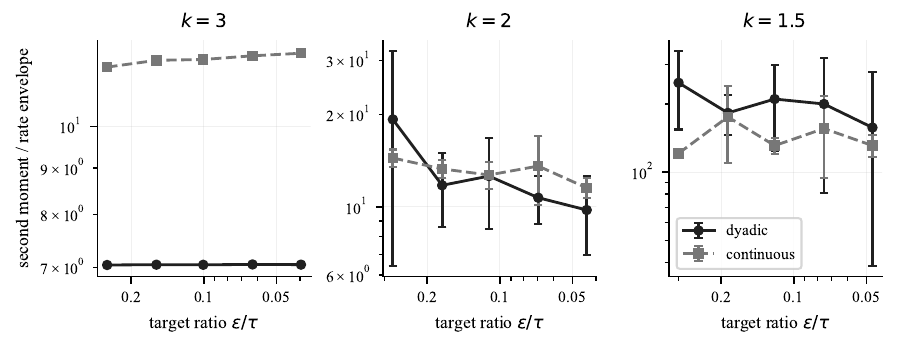}
\caption{Rao--Blackwellized numerical evaluation of the refinement rate.
Each point uses $400{,}000$ independent data draws; vertical bars are $95\%$
normal intervals computed from the conditional second-moment integrand.  The
ordinate is $(\E W_0^2+\E W^2)/(\eps^2\mathfrak v_k)$ for the dyadic
construction and $\E Z_c^2/(\eps^2\mathfrak v_k)$ for the continuous
construction.}
\label{fig:rate-validation}
\end{figure}

Figure~\ref{fig:rate-validation} plots the implemented second moment divided
by $\eps^2\mathfrak v_k(\tau,\eps)$.  The normalized moments remain within
construction-- and regime-specific finite ranges on the tested five-point
grid.  Across the 30 configurations, the largest observed absolute bias was
$0.060\eps$.  The largest paired standardized discrepancy between literal
1-bit simulation and analytical public-coin integration was $1.97$.  Each
pair uses the same observation, and the discrepancy is normalized by the
empirical standard error of the paired differences.

\subsection{Scale-allocation check}

Equation~\eqref{eq:corr-second} and the $k$-moment tail bound reduce the
worst-case correction envelope, up to a $k$-dependent constant, to
\begin{equation}
 \mathcal V(p)=\sum_{j=0}^{J-1}
       \frac{\tau^kL_j^{2-k}}{p_j}.
 \label{eq:allocation-objective}
\end{equation}
Under $\sum_jp_j=1$, Cauchy--Schwarz shows that its minimizer is
$p_j\propto L_j^{(2-k)/2}$, exactly the law in
\eqref{eq:scale-probabilities}.  Figure~\ref{fig:allocation-ablation} evaluates
this deterministic objective for the matched law, the uniform law, and laws
tuned to $k=3$ and $k=1.5$.  The $k=2$ boundary selects the uniform law, while
using the light-tail allocation at $k=1.5$ inflates the envelope by a factor
of $35.5$ in this configuration.

\begin{figure}[!ht]
\centering
\includegraphics[width=.69\linewidth]{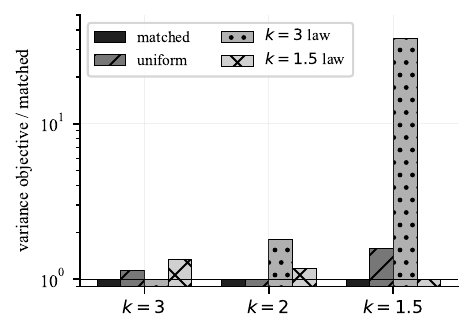}
\caption{Deterministic verification of the scale allocation at
$\sigma=1$, $\mu=0.2$,
$c=0$, and $\eps=0.1$.  Bars show
$\mathcal V(p)/\mathcal V(p^{\star})$ from
\eqref{eq:allocation-objective}; $p^{\star}$ is the matched allocation.
Hatching and gray levels distinguish candidate laws without relying on
color.  The logarithmic ordinate exposes the cost of a regime-mismatched
allocation.}
\label{fig:allocation-ablation}
\end{figure}

\subsection{Numerical consistency and reproducibility}

The mechanism tests verify safe-phase displacement and the pointwise
telescope on $1{,}000$ random $(x,c)$ pairs; the maximum floating-point
residual was $1.42\times10^{-14}$.

The artifact includes a dependency specification, unit tests, raw CSV/JSON
records, plotting code, and deterministic seed schedules.  The experiments
use Python 3.10.11, NumPy 2.2.6, and SciPy 1.15.3 on an ARM macOS 15.5
machine, without a GPU.  Detailed installation, testing, and execution
instructions are provided in the accompanying repository.  The full
validation uses $400{,}000$ data draws for each construction, regime, and accuracy
in Figure~\ref{fig:rate-validation}; all plotted points and interval
half-widths are read directly from the generated CSV file.

\section{Borel representation of the query sets}
\label{app:query-sets}

This section represents each encoder query as a Borel subset of $\R$.  The
query set is fixed before communication, whereas its decoder weight may
depend on the completed transcript.

\subsection{Periodic base and correction queries}

Condition on a realization $(B,U)=(b,u)$ of the public coins of one base
device.  Its query is the indicator of the set
\begin{equation}
\label{eq:direct-base-set}
 \begin{aligned}
 A^{(0)}_{b,u}
 &=\{x\in\R:\rho_{L_0,b}(x)\geq u\}\\
 &=\bigcup_{m\in\mathbb Z}
 \left[\frac{bL_0}{2}+mL_0+u,
       \frac{bL_0}{2}+(m+1)L_0\right).
 \end{aligned}
\end{equation}
Thus the base bit is simply $Y=\1\{X\in A^{(0)}_{b,u}\}$ for a countable
union of half-open intervals fixed by the public seed.

Likewise, conditional on $(K,B,B',U)=(j,b,b',u)$, a correction device uses
\begin{equation}
\label{eq:direct-correction-set}
 A^{(1)}_{j,b,b',u}
 =\{x\in\R:f_{j,b,b'}(x)\geq u\}.
\end{equation}
Both residue functions in \eqref{eq:fj-pair} are right-continuous and Borel
measurable.  Moreover, $\rho_{L_j,b}$ has period $L_j$, whereas
$\rho_{L_{j+1},b'}$ has period $L_{j+1}=2L_j$.  Their difference is therefore
$2L_j$-periodic and piecewise constant on the locally finite partition formed
by the two boundary grids.  Consequently,
$A^{(1)}_{j,b,b',u}$ is a countable union of half-open intervals and is Borel.
Neither set definition depends on the center $c$.

For the continuous alternative, condition on $(R,U,\xi)$.  The encoder set
is exactly the periodic coloring set in \eqref{eq:alt-query-set}, again a
countable union of half-open intervals.  The infinite sign sequence is merely
a convenient representation of a public random Borel set.  The calculation
in \eqref{eq:alt-rademacher-identity} requires only pairwise independence of
the cell colors.

\begin{proposition}[Fixed Borel implementation]
\label{prop:fixed-borel}
Condition on all public seeds drawn before sampling.  Every encoder in the
discrete-scale protocol and in the continuous alternative applies a
deterministic indicator of a Borel subset of $\R$.  The collection
of these sets is independent of the localization transcript.  Hence both
protocols are fully non-adaptive in the sense of
Definition~\ref{def:nonadaptive}.
\end{proposition}

\begin{proof}
The localization sets are fixed by Proposition~\ref{prop:localization}.
Equations \eqref{eq:direct-base-set} and
\eqref{eq:direct-correction-set} establish the claim for the two refinement
blocks of the discrete construction, while \eqref{eq:alt-query-set}
establishes it for the continuous construction.  Their seeds are sampled
jointly before any data are observed and contain no transcript-dependent
quantity.  Conditional on the seeds, the transmitted bit is therefore
$\1\{X_i\in A_i\}$ for a fixed Borel set $A_i$.  Since the refinement seeds
and samples are also independent of the localization block, the collection
$(A_i)_i$ is independent of the decoded center and of the entire localization
transcript.
\end{proof}

\subsection{Separation of decoder and query adaptation}

Let $\mathcal F_{\mathrm{pub}}$ be the sigma-field generated by all public
coins and let $\mathcal F_{\mathrm{loc}}$ be generated by the localization
bits.  Each refinement query set $A_i$ is
$\mathcal F_{\mathrm{pub}}$-measurable, whereas
$c$ is measurable with respect to
$\mathcal F_{\mathrm{pub}}\vee\mathcal F_{\mathrm{loc}}$.  The selected safe
phases $b_j(c)$, the subtracted threshold indicators in
\eqref{eq:base-stat} and \eqref{eq:correction-stat}, and the gate and control
variate in \eqref{eq:alt-statistic} are all decoder-side functions evaluated
only after the bits have arrived.  They can change the numerical weight
assigned to a bit, including setting it to zero, but cannot change the event
$\{X_i\in A_i\}$ whose truth value the device already transmitted.

This is the same distinction used in randomized sketching: a sketch may be
drawn and applied obliviously, while its reconstruction map depends on side
information learned later.

\subsection{Scope of the result}

The theorem is information-theoretic and one-dimensional.  The public query
plan is compiled with $(k,\lambda,\sigma,\eps,\delta)$ known and uses shared
randomness, as stated in Definition~\ref{def:nonadaptive} and the protocol
summary in Section~\ref{sec:protocol-summary}.  Its arbitrary measurable
query sets are
generally nonlocal, countable unions of intervals; consequently, the theorem
does not remove the lower bounds for non-adaptive threshold and interval
queries.

Finally, a coordinate-wise extension to multivariate means would generally
fail to capture high-dimensional geometry and need not preserve the optimal
dependence on dimension.  Extending universal decoder-side refinement beyond
the scalar setting therefore requires additional ideas.

\end{document}

%% file: fig_main_mechanism.tex
\begin{tikzpicture}[
  x=1cm,
  y=1cm,
  font=\sffamily\footnotesize,
  gridline/.style={draw=black!60,line width=.65pt},
  boundary/.style={draw=black!54,dashed,line width=.60pt},
  centerline/.style={draw=black!62,densely dotted,line width=.60pt}
]
  % Two period-L boundary grids, shifted by L/2.
  \node[anchor=east,text=black!70] at (1.40,2.08)
    {$\mathcal G_{L,0}$};
  \node[anchor=east,text=black!88,font=\sffamily\footnotesize\bfseries]
    at (1.40,.90) {$\mathcal G_{L,1}$ (safe)};
  \draw[gridline] (1.58,2.00) -- (8.18,2.00);
  \draw[gridline,line width=.90pt] (1.58,.82) -- (8.18,.82);

  \foreach \x in {1.58,3.78,5.98,8.18} {
    \draw[boundary] (\x,1.67) -- (\x,2.33);
  }
  \foreach \x in {2.68,4.88,7.08} {
    \draw[boundary] (\x,.49) -- (\x,1.15);
  }

  % Choose c on a phase-0 boundary.  The half-shifted phase then has a
  % boundary-free L/4 neighborhood around c.
  \fill[black!10,rounded corners=.8pt] (3.23,.45) rectangle (4.33,1.19);
  \draw[centerline] (3.78,.28) -- (3.78,2.45);
  \fill[black!88] (3.78,2.00) circle (1.8pt);
  \fill[black!88] (3.78,.82) circle (1.8pt);
  \node[above=2pt,font=\sffamily\footnotesize\bfseries] at (3.78,2.42) {$c$};
  \node[below=2pt,text=black!76,font=\sffamily\scriptsize]
    at (3.78,.38) {$[c-L/4,\,c+L/4]$};
\end{tikzpicture}

%% file: fig_random_grid_kernel.tex
\begin{tikzpicture}[
  x=1cm,
  y=1cm,
  font=\sffamily\footnotesize,
  gridline/.style={draw=black!60,line width=.65pt},
  boundary/.style={draw=black!54,dashed,line width=.60pt},
  centerline/.style={draw=black!62,densely dotted,line width=.60pt},
  axis/.style={-{Latex[length=1.7mm,width=1.1mm]},draw=black!62,
    line width=.65pt},
  note/.style={font=\sffamily\scriptsize,text=black!72}
]
  % (a) Decoder-side signed grid.
  \node[anchor=west,font=\sffamily\footnotesize\bfseries,text=black!82]
    at (.22,5.72) {(a)};

  \foreach \x/\lab/\shade in {
      .38/Q-1/black!6,
      2.28/Q/black!14,
      4.18/Q+1/black!6,
      6.08/Q+2/black!3
    } {
      \draw[gridline,fill=\shade] (\x,4.20) rectangle +(1.90,.88);
      \node[text=black!74] at (\x+.95,4.63) {$\xi_{\lab}$};
    }

  % The central half of the Q-cell is exactly the gate A=1.
  \fill[black!35,opacity=.30] (2.755,4.20) rectangle (3.705,5.08);
  \draw[boundary] (2.755,4.02) -- (2.755,5.23);
  \draw[boundary] (3.705,4.02) -- (3.705,5.23);
  \draw[decorate,decoration={brace,amplitude=3.2pt,mirror},draw=black!64]
    (2.755,4.10) -- node[note,below=3.5pt] {$A=1$} (3.705,4.10);

  \draw[centerline] (3.22,4.92) -- (3.22,5.45);
  \fill[black!88] (3.22,5.30) circle (1.65pt);
  \node[above=2pt,font=\sffamily\scriptsize\bfseries] at (3.22,5.30) {$c$};
  \draw[centerline] (5.08,4.92) -- (5.08,5.45);
  \fill[white,draw=black!82,line width=.65pt] (5.08,5.30) circle (1.75pt);
  \node[above=2pt,font=\sffamily\scriptsize\bfseries,text=black!80]
    at (5.08,5.30) {$X$};
  \draw[-{Latex[length=1.7mm]},draw=black!70,line width=.75pt]
    (3.32,5.58) -- node[above=1pt,note] {$d=X-c$} (4.98,5.58);

  \node[note] at (1.33,3.47) {$-1$};
  \node[note] at (3.23,3.47) {$0$};
  \node[note] at (5.13,3.47) {$+1$};
  \node[note] at (7.03,3.47) {$0$};
  \node[note,anchor=west] at (.72,3.15)
    {$\E_\xi[(\widetilde B-\xi_Q)(\xi_{Q+1}-\xi_{Q-1})]$};

  % (b) The compact trapezoidal kernel psi_{1/4}.
  \node[anchor=west,font=\sffamily\footnotesize\bfseries,text=black!82]
    at (.22,2.64) {(b)};
  \coordinate (org) at (.68,.52);
  \draw[axis] (org) -- (7.72,.52)
    node[below=1pt,note] {$t=|d|/r$};
  \draw[axis] (org) -- (.68,2.28);
  \node[note,anchor=west] at (.84,2.16) {$\psi_{1/4}(t)$};

  \fill[black!10]
    (1.36,.52) -- (2.72,1.65) -- (4.08,1.65) -- (5.44,.52) -- cycle;
  \draw[draw=black!78,line width=.90pt]
    (.68,.52) -- (1.36,.52) -- (2.72,1.65) --
    (4.08,1.65) -- (5.44,.52) -- (7.38,.52);
  \foreach \x/\lab in {
      1.36/$\frac14$,
      2.72/$\frac34$,
      4.08/$\frac54$,
      5.44/$\frac74$
    } {
      \draw[boundary] (\x,.46) -- (\x,1.79);
      \draw[black!55] (\x,.46) -- (\x,.58);
      \node[below=3pt,font=\sffamily\scriptsize] at (\x,.52) {\lab};
    }
  \draw[black!55] (.62,1.65) -- (.74,1.65);
  \node[left=2pt,font=\sffamily\scriptsize] at (.68,1.65) {$\frac12$};
  \node[note,anchor=west] at (3.08,2.16)
    {$C_{1/4}=\int_0^\infty\!\psi_{1/4}(t)t^{-2}dt=\log(15/7)$};
\end{tikzpicture}